\documentclass[conference,letterpaper]{IEEEtran}

\IEEEoverridecommandlockouts

\usepackage[utf8]{inputenc} 
\usepackage[T1]{fontenc}
\usepackage{url}
\usepackage{ifthen}
\usepackage{cite}
\usepackage[cmex10]{amsmath} 

\usepackage{cite}
\usepackage{amsmath,amssymb,amsfonts}
\usepackage{algorithmic}
\usepackage{graphicx}
\usepackage{textcomp}
\usepackage{xcolor}
\usepackage{balance}
\usepackage{cleveref}
\def\BibTeX{{\rm B\kern-.05em{\sc i\kern-.025em b}\kern-.08em
    T\kern-.1667em\lower.7ex\hbox{E}\kern-.125emX}}

\usepackage{tikz}

\usetikzlibrary{arrows.meta, bending}

\usepackage{tikzscale}
\usepackage{pgfplots}
\pgfplotsset{compat=1.18} 
\usepgfplotslibrary{fillbetween}

\definecolor{tabblue}{HTML}{1f77b4}
\definecolor{taborange}{HTML}{ff7f0e}
\definecolor{tabgreen}{HTML}{2ca02c}
\definecolor{tabred}{HTML}{d62728}
\definecolor{tabpurple}{HTML}{9467bd}
\definecolor{tabbrown}{HTML}{8c564b}
\definecolor{tabpink}{HTML}{e377c2}
\definecolor{tabgray}{HTML}{7f7f7f}
\definecolor{tabolive}{HTML}{bcbd22}
\definecolor{tabcyan}{HTML}{17becf}
\definecolor{deepblue}{HTML}{0000ff}

\usepackage{amsthm}
\usepackage{mathtools}
\usepackage{bbm}

\newtheorem{theorem}{Theorem} 
\newtheorem{lemma}{Lemma}

\newtheorem{corollary}{Corollary}

\theoremstyle{definition}
\newtheorem{definition}{Definition}

\DeclarePairedDelimiterX{\infdivx}[2]{(}{)}{%
  #1\;\delimsize\|\;#2%
}
\DeclarePairedDelimiterX{\infdivmi}[2]{(}{)}{%
  #1;#2%
}
\DeclarePairedDelimiterX{\infdivcsmi}[2]{(}{)}{%
  #1 \to #2%
}

\newcommand{\KLD}{D_{\mathrm{KL}}\infdivx}
\newcommand{\MI}{I\infdivmi}
\newcommand{\CSD}{D_{\mathrm{CS}}\infdivx}

\newcommand{\qty}[1]{\left\{ #1 \right\}}
\newcommand{\norm}[1]{\left\lVert #1 \right\rVert}
\newcommand{\RND}[2]{\frac{d #1}{d #2}}
\newcommand{\Ind}[1]{\mathbf{1}\!\left[ #1 \right]}
\newcommand{\Normal}[2]{\mathcal{N}(#1, #2)}
\newcommand{\FI}{I_\mathrm{F}\infdivcsmi}

\newcommand{\Nats}{\mathbb{N}}

\DeclareMathOperator{\enc}{\mathrm{enc}}
\DeclareMathOperator{\dec}{\mathrm{dec}}
\DeclareMathOperator{\XSpace}{\mathcal{X}}
\DeclareMathOperator{\YSpace}{\mathcal{Y}}
\DeclareMathOperator{\ZSpace}{\mathcal{Z}}
\DeclareMathOperator{\Exp}{\mathbb{E}}
\DeclareMathOperator{\strings}{\{0, 1\}^*}
\DeclareMathOperator{\Prob}{\mathbb{P}}
\DeclareMathOperator{\Ent}{\mathnormal{H}}
\DeclareMathOperator{\DiffEnt}{\mathnormal{h}}
\DeclareMathOperator{\defeq}{\,\triangleq\,}

\DeclareMathOperator{\Unif}{\mathrm{Unif}}
\DeclareMathOperator{\Geom}{\mathrm{Geom}}

\DeclareMathOperator{\Oh}{\mathcal{O}}

\newcommand{\CrossEnt}{\Ent\infdivmi}
\newcommand{\filteredInds}{\mathcal{S}}

\usepackage[final]{changes}

\usepackage[top=0.75in, bottom=1.05in, left=0.625in, right=0.625in]{geometry}
\begin{document}

\title{Rejection Sampling is Optimal for \\
Relative Entropy Coding
\thanks{The authors acknowledge financial support from Imperial College London through an Imperial College Research Fellowship grant awarded to GF. SH, FA, and TL acknowledge support from the Natural Sciences and Engineering Research Council of Canada.}
}

\author{\IEEEauthorblockN{Spencer Hill\textsuperscript{*}}\thanks{\textsuperscript{*}Work initiated while visiting GF at Imperial College London.}
\IEEEauthorblockA{Queen's University \\
Kingston, Canada \\
spencer.hill@queensu.ca}
\and
\IEEEauthorblockN{Fady Alajaji}
\IEEEauthorblockA{Queen's University \\
Kingston, Canada \\
fa@queensu.ca}
\and
\IEEEauthorblockN{Tam\'as Linder}
\IEEEauthorblockA{Queen's University \\
Kingston, Canada \\
tamas.linder@queensu.ca}
\and
\IEEEauthorblockN{Gergely Flamich}
\IEEEauthorblockA{Imperial College London \\
London, UK \\
g.flamich@imperial.ac.uk} 
}

\maketitle

\begin{abstract}
In relative entropy coding, a sender aims to design a stochastic code such that, on input {\boldmath $X \sim P_X$}, the receiver can generate a sample {\boldmath $Y \sim P_{Y \mid X}$}. 
It is a standard result that (1) this requires at least {\boldmath $\MI{X}{Y}$} bits, (2) the lower bound is achievable within a logarithmic gap, and (3) this gap cannot be reduced in general. 
The necessity of the gap suggests that mutual information is not the appropriate information measure for quantifying the rate of relative entropy coding.
\par
A potential alternative emerged in the work of Flamich et al.~\cite{flamich2025redundancy}, who proved a tighter lower bound of {\boldmath $\FI{X}{Y}$}, a quantity we call the \textit{functional information}.
In this paper, we \deleted{show this lower bound is tight by constructing} \added{construct} the \emph{ring toss code}, an encoding method for rejection sampling which uses at most {\boldmath $\FI{X}{Y} + \log e$ bits}, \added{hence obtaining a tight one-shot characterization of relative entropy coding.} 
\par
For the trivial channel \boldmath{$Y = X$}, our result recovers the noiseless source coding theorem within a small constant.
For a general channel, it implies that the classical mutual information lower bound is achievable within {\boldmath ${\log(\MI{X}{Y} + 1) + 2.45}$} bits in general and within {\boldmath $1.45$} bits for singular channels, which are both the tightest bounds of their kind to date. 
Moreover, our one-shot result also recovers Sriramu and Wagner's asymptotic results on the second-order redundancy of relative entropy codes.
\end{abstract}

\begin{IEEEkeywords}
Relative entropy coding, channel simulation, rejection sampling, ring toss code, singular channels. 
\end{IEEEkeywords}

\section{Introduction}
For a pair of random variables $X, Y \sim P_{X, Y}$, the goal of relative entropy coding (also called channel simulation) is to find a \emph{stochastic code}~\cite{flamich2026data} that, given an input $X \sim P_X$, encodes a single sample $Y \sim P_{Y \mid X}$ in finitely many bits. In doing so, the sender and receiver \emph{simulate} the channel ${X \to Y}$. 
Formally, a stochastic code is a triplet $(Z, \enc, \dec)$ where
\begin{itemize}
    \item $Z$ is a random variable independent of $X$ shared by encoder and decoder called the \emph{common randomness};
    \item $\enc: \XSpace \times \ZSpace \to \strings$ is the \emph{encoder} mapping an input and common randomness realization to a binary string;
    \item $\dec : \strings \times \ZSpace \to \YSpace$ is the \emph{decoder} mapping the encoder output and common randomness to a point $y \in \YSpace$. 
\end{itemize}
Together with the correctness requirement that for all $x \in \XSpace$, 
\begin{equation}
    \dec(\enc(x, Z), Z) \sim P_{Y \mid X = x},
\end{equation}
it can be shown that~\cite{harsha2010communication, li2018strong} 
\begin{equation}
\label{eq:basic_mi_lower_bound}
\MI{X}{Y} \leq H(Y \mid Z) \leq \Exp [\lvert \enc(X, Z) \rvert],
\end{equation}
where $\lvert \, \cdot \, \rvert$ is the length of a binary string. There is a long line of research work~\cite{harsha2010communication, braverman2014public, li2018strong, flamich2023adaptive} to design stochastic codes which almost achieve this lower bound, the crowning achievement of which is the \emph{Poisson functional representation}~\cite{li2018strong}, yielding
\begin{equation}
    H(Y \mid Z) \leq \MI{X}{Y} + \log (\MI{X}{Y} + 1) + 4. 
\end{equation}
This logarithmic redundancy over the mutual information is also achieved by other schemes, including \emph{greedy rejection sampling}~\cite{harsha2010communication, flamich2023adaptive} and \emph{greedy Poisson rejection sampling}\cite{flamich2023greedy}. However, it is not possible to remove this $\log$ term, as there exist channels for which $\MI{X}{Y} + \log ( \MI{X}{Y} + 1) - 1$ bits are necessary~\cite{braverman2014public, li2018strong}. 
Thus, a natural question is whether a different information measure and encoding scheme could yield a tight bound. 
\par
In~\cite{flamich2025redundancy}, progress was made in this direction by improving the lower bound on any relative entropy code to
\begin{equation}
\label{eq:csmi_lower_bound}
\FI{X}{Y} \leq H(Y \mid Z),
\end{equation}
where $\FI{X}{Y}$ is the \emph{functional information} (see Definitions~\ref{def:csd} and~\ref{def:csmi}). Note that $\MI{X}{Y} \leq \FI{X}{Y}$, so \eqref{eq:csmi_lower_bound} is at least as tight as \eqref{eq:basic_mi_lower_bound}.
However, the lower bound in~\eqref{eq:csmi_lower_bound} has only been shown to be achievable (within $\log e$ bits) for additive exchangeable noise channels with finite field input~\cite{zhao2025rejection}. In this paper, we construct the \emph{ring toss code}, an encoding method for standard rejection sampling. 
We prove for general channels that it achieves the lower bound in~\eqref{eq:csmi_lower_bound} within an additive constant of $\log e$, demonstrating that \textbf{rejection sampling is optimal for relative entropy coding}.  
\par 
Concretely, our contributions are as follows:
\begin{itemize}
\item We construct the ring toss code and prove that its one-shot rate satisfies ${H(Y \mid Z) \leq \FI{X}{Y} + \log e}$. This upper bound matches the lower bound in~\eqref{eq:csmi_lower_bound} within an additive constant of $\log e \approx 1.44$. The algorithm and proof are remarkably simple, relying entirely on a changed encoding of the index returned by standard rejection sampling. 
\item For a class of channels called \emph{singular channels}~\cite{altuug2014refinement, sriramu2024optimal} (see Definition~\ref{def:singular}), we prove that $\FI{X}{Y} = \MI{X}{Y}$. Therefore, our work identifies a large class of channels for which we can achieve constant non-asymptotic redundancy over the mutual information. 
\item For general channels, our result implies the tighter achievability result in terms of the mutual information 
\begin{align}
\Ent(Y \mid Z) < \MI{X}{Y} + \log(\MI{X}{Y} + 1) + 2.45.
\end{align}
\item Finally, our one-shot result recovers the asymptotic characterization of relative entropy coding in~\cite{sriramu2024optimal, flamich2025redundancy}, and is achieved by a simpler scheme with optimal redundancy.
\end{itemize}
\subsection{Notation}
In this paper, $\log$ denotes the binary logarithm, and all information measures are in bits. 
For an event $A$, its indicator function is $\Ind{A}$.
Unless otherwise stated, all random variables are assumed to take values in some arbitrary Polish space.
Random variables are written in capital letters $X$, we write $X \sim P_X$ to say that $X$ has probability measure $P_X$, and use $P_X^{\otimes n}$ for $n$ independent copies of $P_X$. Random variable alphabets or spaces are denoted by caligraphic letters $X \in \XSpace$, $Y \in \YSpace$. 
We denote the expectation of $X$ and the conditional expectation of $X$ given $Y$ as $\Exp[X]$ and $\Exp[X \mid Y]$, respectively.
For two probability measures $P$ and $Q$, we write $P \ll Q$ to indicate that $P$ is absolutely continuous with respect to $Q$, in which case $\RND{P}{Q}$ denotes the Radon-Nikodym derivative. 
For a discrete probability distribution $P$, we denote its Shannon entropy as $\Ent(P)$.
For a distribution $P$ over the reals with density function $f$, we denote its differential entropy as $\DiffEnt(f)$ or $\DiffEnt(P)$. We write $\CrossEnt{P}{Q}$ for the cross-entropy between two distributions $P$ and $Q$. For a $Q$-measurable function $r$, we denote its $Q$-essential supremum as $\norm{r}_\infty$.

\section{Preliminaries} 
We begin by defining the channel simulation divergence. We refer the reader to~\cite{flamich2024data} for a comprehensive overview of the channel simulation divergence and its properties. 
\begin{definition}[Width Function] For two probability measures $P \ll Q$, the $Q$ width function at height level $\ell$ is defined as
\begin{equation}
\label{eq:widthfunction} 
    w_Q(\ell) \defeq \Prob_{Z \sim Q} \left( \RND{P}{Q} (Z) \geq \ell \right) = Q \left( \RND{P}{Q} \geq \ell \right)\!. 
\end{equation}
\end{definition}
\added{We can analogously define $w_P(\ell)$ as the $P$ width of $\RND{P}{Q}$ at height $\ell$.} As visualized in Fig.~\ref{fig:width}, $w_Q(\ell)$ measures the $Q$-width of the superlevel set of $\RND{P}{Q}$ at height $\ell$. 
It can be shown that $w_Q(\ell)$ is a probability density function (pdf), and the channel simulation divergence is defined as its differential entropy. 
\begin{definition}[Channel Simulation Divergence \cite{goc2024channel}] \label{def:csd} 
Let ${P \ll Q}$ with $Q$-width function $w_Q$. Then, the channel simulation divergence of $P$ from $Q$ is 
\begin{equation}
\CSD{P}{Q}
\defeq -\int_{0}^\infty w_Q(\ell) \log w_Q(\ell) \, d\ell
=\DiffEnt(w_Q). \label{eq:csd}
\end{equation}
\end{definition}
\begin{figure}
    \centering
    \includegraphics[]{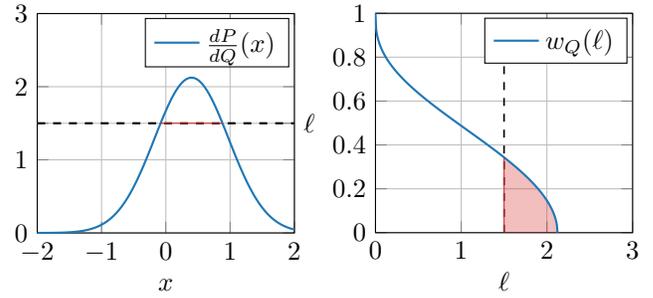}
    \caption{Visualization of the width function when $P = \Normal{0.3}{0.25}$ and $Q = \Normal{0}{1}$. \textbf{Left:} the density ratio $\RND{P}{Q}$; the value of the width function at $\ell = 1.5$ is the $Q$-width of the \textcolor{tabred}{red} line. \textbf{Right:} the width function $w_Q(\ell)$, with the value at $\ell = 1.5$ also marked as in the left plot.}
    \label{fig:width}
\end{figure}
Lemma IV.1 of \cite{goc2024channel} establishes that $D_{\mathrm{CS}}$ is a divergence: it is nonnegative and is zero if and only if $P = Q$. 
Additionally, for $\kappa = \KLD{P}{Q}$ the Kullback-Leibler (KL) divergence,
\begin{equation}
    \kappa \leq \CSD{P}{Q} \leq \kappa + \log (\kappa + 1) + 1. \label{eq:kltocsd}
\end{equation}
\par
In a slight abuse of notation, for random variables ${X, Y \sim P_{X, Y}}$ and $y \in \YSpace$ we write 
\begin{equation}
\label{eq:ywidth}
w_y(\ell) \defeq \Prob_{X \sim P_X} \left( \RND{P_{X \mid Y}}{P_X}(X \mid y) \geq \ell \right). 
\end{equation}
This way, we can write the channel simulation divergence as $\CSD{P_{X \mid Y = y}}{P_X} = -\int_0^\infty w_y(\ell) \log w_y(\ell) \,d\ell$. \par
Next, generalizing a result of Li and El Gamal \cite{li2018strong}, Flamich et al.~\cite{flamich2025redundancy} showed that for any stochastic code 
\begin{equation}
    \Exp[\CSD{P_{X \mid Y}}{P_X}] \leq \Ent(Y \mid Z). \label{eq:lowerbound}
\end{equation}
This motivates the following definition, whose name is inspired by the \emph{strong functional representation lemma}~\cite{li2018strong}, and which is a generalization of the ``excess functional information lower bound'' \cite{li2018strong,zhao2025rejection}.
\begin{definition}[Functional Information]  
\label{def:csmi}
Let $X, Y \sim P_{X, Y}$ be a pair of dependent random variables.
For $y \in \YSpace$, let $w_y$ be as in \eqref{eq:ywidth}.
Then, the \emph{functional information of $X$ to $Y$} is
\begin{equation}
\FI{X}{Y} \defeq \Exp[\CSD{P_{X \mid Y}}{P_X}] = \Exp[\DiffEnt[w_Y]]. \label{eq:csmi}
\end{equation}
Then,~\eqref{eq:lowerbound} may be restated as $\FI{X}{Y} \leq H(Y \mid Z)$. 
\end{definition}
\subsection{Rejection Sampling}
We briefly describe the rejection sampling algorithm, which forms the basis of the ring toss code; see~\cite{flamich2026data, li2024channel} for a detailed treatment. 
Rejection sampling simulates samples from a \textit{target distribution} $P$ given independent and identically distributed (i.i.d.) samples $\qty{Y_i}_{i \geq 1}$ from a \textit{proposal distribution} $Q$. Concretely, let $r = dP / dQ$ and assume that $r$ is uniformly bounded by a constant $M$. Then, rejection sampling sequentially examines the proposal samples, at each step $i$ accepting sample $Y_i$ with probability $r(Y_i) / M$. 
Thus, for $\qty{U_i}_{i \geq 1}$ i.i.d. $\Unif(0, 1)$, the uniform distribution on $(0, 1)$, rejection sampling returns the index
\begin{equation}
K \defeq \min \qty{i : U_i \leq \frac{1}{M} r(Y_i)}. \label{eq:rejsampling}
\end{equation}
Under this definition, it is a standard result that $Y_K \sim P$. 
Now, we can use rejection sampling to construct a stochastic code as follows. 
We begin by assuming that there is some bound 

\noindent $M$ such that $dP_{Y \mid X} / dP_Y \, (y \mid x) \leq M$ for all $x$ and $y$.
Upon input $X \sim P_X$, we set $Q \gets P_Y$ and $P \gets P_{Y \mid X}$. 
Here, the shared randomness is ${Z = \qty{(Y_i, U_i)}_{i \geq 1} \overset{\text{i.i.d.}}{\sim} P_Y \times \Unif(0, 1)}$ and the encoder is a prefix-free encoding of the index $K$, from which the decoder can recover $Y_K \sim P_{Y \mid X = x}$. 
One can observe that $K \sim \Geom(1 / M)$, but using this unconditional distribution yields a rate of
\begin{align}
\Ent(Y \mid Z) \leq \Ent(K) \leq \log M + 1.  
\end{align}
The issue is that $\log M$ can be much larger than $\MI{X}{Y}$.
Essentially, this encoding is suboptimal because the unconditional distribution of $K$ discards the information the decoder has from the shared randomness $Z$. 
Our main contribution is demonstrating that, by considering the distribution of $(K \mid Z)$, we can obtain a tight upper bound on $H(Y \mid Z)$. This is conceptually similar to the construction of~\cite{phan2025channel}, which uses sorting to adapt the encoding of rejection sampling such that it uses at most $\MI{X}{Y} + \log (\MI{X}{Y} + 1) + 9$ bits. 
\par
The ring toss code leverages the fact that the decoder can compute the acceptance probability of each proposal sample $Y_k$ under $P_X$ given its shared access to $Z$. This leads to a more efficient coding distribution and hence a tighter bound on $H(Y \mid Z)$. Moreover, the probability computed is exactly a width function, which is why the channel simulation divergence appears in the upper bound. 
\section{The Ring Toss Code}
\begin{theorem} \label{thm:rts}
Let $X, Y \sim P_{X, Y}$ be dependent random variables such that $\RND{P_{Y \mid X}}{P_Y}(y \mid x) \leq M$ for all $x$ and $y$. 
Let ${Z = \qty{(Y_i, U_i)}_{i \geq 1} \overset{\text{i.i.d.}}{\sim} P_Y \times \Unif(0, 1)}$ be the common randomness shared between encoder and decoder and define $K$ to be the index returned by rejection sampling, 
\begin{align*}
K(X, Z) \defeq \added{\min} \qty{\!i \in \Nats : U_i \leq \!\frac{1}{M} \frac{dP_{Y \mid X}}{dP_Y}(Y_i \!\mid\! X)\!}\!.
\end{align*}
Then, $Y_K \sim P_{Y \mid X}$ and 
\begin{equation}
\Ent(K \mid Z) \leq \FI{X}{Y} + \log e.
\end{equation}
Moreover, this upper bound is achieved by encoding $K$ using a prefix-free code for the distribution
\begin{equation}
    Q_{K \mid Z}(k \mid z) \defeq w_{y_k}(M u_k) \prod_{i=1}^{k-1} (1 - w_{y_i}(M u_i)), \label{eq:q}
\end{equation}
where $w_{y}$ is the width function defined as in~\eqref{eq:ywidth} and ${z = \qty{(y_i, u_i)}_{i \geq 1}}$ is an instance of the common randomness. 
\end{theorem}
\begin{proof}
We show in the online appendix~\cite{hill2026rejection} that $Q_{K \mid Z}$ is a valid probability distribution. 
Then, letting ${(K \mid Z) \sim P_{K \mid Z}}$, and noting that $\Ent(Y \mid Z) \leq \Ent(K \mid Z) \leq \CrossEnt{P_{K \mid Z}}{Q_{K \mid Z}}$, we prove Theorem~\ref{thm:rts} by bounding the cross-entropy.
First,
\begin{align}
&\CrossEnt{P_{K \mid Z}}{Q_{K \mid Z}} = \Exp[- \log Q_{K \mid Z}(K \mid Z)] \nonumber \\
&= \Exp[-\log w_{Y_K}(MU_K) ] - \Exp\left[ \sum_{i=1}^{K-1} \log(1 - w_{Y_i}(MU_i)) \right]\!\!. \label{eq:cross_entropy_terms}
\end{align}
We will address each term in~\eqref{eq:cross_entropy_terms} separately. 
However, before doing so, let us determine the joint distribution of $(Y_K, MU_K)$. 
Let ${r(y \mid x) \defeq \RND{P_{Y \mid X}}{P_Y}}(y \mid x)$ and note that by the definition of rejection sampling, we have
\begin{align}
\label{eq:rs_accept_coord_joint_distribution}
\frac{dP_{Y_K, U_K \mid X}}{d(P_Y \!\times\! \Unif(0, 1))}(y,\! u\! \mid\! X) = M \!\cdot\!\Ind{M u \leq r(y \mid X)}\!.
\end{align}
Note that since we are using a universal upper bound $M$,  we have $(K \mid X) \sim \Geom(1 / M)$ irrespective of $X$. 
In particular, ${\Prob(K \geq k \mid X) = \left(1 - \frac{1}{M} \right)^{k-1}}$ and $\Exp[K \mid X] = M$.
Finally, letting $(Y_1, U_1) \sim P_Y \times \Unif(0, 1)$, we get
\begin{align}
&\Exp[- \log\,w_{Y_K}(M U_K)] \nonumber \\
&= \Exp\left[\Exp[-\log w_{Y_K}(MU_K) \mid X] \right] \nonumber \\
&= -\Exp\left[M\Ind{MU_1 \!\leq\! r(Y_1 \mid X)} \log w_{Y_1}(MU_1)\right] \tag{by \cref{eq:rs_accept_coord_joint_distribution}} \\
&= -\Exp\left[M\Exp[\Ind{MU_1 \leq r(Y_1 \mid X)} \mid Y_1, U_1] \cdot \log w_{Y_1}(MU_1)\right] \nonumber \\
&= -\Exp\left[M\cdot w_{Y_1}(MU_1)\log w_{Y_1}(MU_1)\right] \tag{def.\ of $w_y(\ell)$}\\
&= -\Exp\left[\int_0^1 M\cdot w_{Y_1}(Mu)\log w_{Y_1}(Mu) \, du\right] \nonumber \\
&= \Exp[\CSD{P_{X \mid Y}}{P_X}] \tag{sub.\ $\ell \gets Mu$}\\
&= \FI{X}{Y}, \nonumber
\end{align}
Next, we deal with the second term in~\eqref{eq:cross_entropy_terms}.
First, we let ${R_i = \log (1 - w_{Y_i}(MU_i)) \cdot(1 - \Ind{MU_i \leq r(Y_i \!\mid\! X)})}$ and get
\begin{align}
\Exp \Bigg[ &\sum_{i=1}^{K-1} \log (1 - w_{Y_i}(MU_i)) \Bigg] \nonumber \\
&= \Exp \left[\Exp \left[ \sum_{i=1}^K R_i \,\Big\lvert \, X \right]\right] \nonumber \\
&= \Exp \left[ \Exp \left[ K \mid X \right]\Exp \left[ R_1 \mid X \right]  \right]\label{eq:wald} \\
&= M \Exp \big[  \big(1 - \Exp\left[\Ind{MU_1 \!\leq\! r(Y_1 \mid X)} \mid Y_1, U_1 \right]\big) \,\cdot \nonumber\\
&\quad\quad\quad\quad\quad\cdot\log(1 - w_{Y_1}(MU_1)) \big] \nonumber \\
&= M \Exp \left[ (1 - w_{Y_1}(MU_1)) \log (1 - w_{Y_1}(MU_1)) \right] \label{eq:defofwidth} \\
&\geq M \Exp\left[-\int_0^1 w_{Y_1}(Mu) \log e \, du \right]  \label{eq:logineq}\\
&= - \log e, \tag{sub.\ $\ell \gets Mu$ and $w_y$ is a pdf}
\end{align}
where~\eqref{eq:wald} follows from the generalized Wald's equation, since conditioned on $X$, $K$ is a stopping time adapted to $\{R_i\}_{i \ge 1}$, and~\eqref{eq:defofwidth} follows from the definition of $w_y$. \Cref{eq:logineq} is an application of the law of total expectation and the inequality $(1 - x) \log (1 - x) \geq -x \log e$. 
Theorem~\ref{thm:rts} follows. 
\end{proof}
The ring toss code has a nice geometric interpretation: it transforms relative entropy coding into a search problem over $\XSpace$. 
Conventionally, a relative entropy code searches for a sample $Y_K$ which falls under the Radon-Nikodym derivative (or satisfies an analogous property). In this way, most relative entropy codes search over $\YSpace$. Instead, the ring toss code searches over $\XSpace$ for the set such that $M U_i \leq r(Y_i \mid X)$, as shown in Fig.~\ref{fig:rts}. \par 
Defining the sets $S_i = \qty{x \in \mathcal{X} : \RND{P_{Y \mid X}}{P_Y}(Y_i \mid x) \geq M U_i}$, it is actually possible to exactly compute the distribution $P_{K \mid Z}$. In particular, $K = k$ if (1) $x \in S_k$, and (2) $x \notin S_i$ for each $i = 1, \ldots, k-1$. Hence, 
\begin{align}
P_{K \mid Z}(k \mid z) = P_X\left(S_k \setminus \bigcup_{i=1}^{k-1} S_i\right).
\end{align}
One can view $Q_{K \mid Z}$ as an approximation of this distribution, which would be exact if the supports $S_i$ had measure zero overlap under $P_X$. However, \added{in light of the functional information lower bound,} using this more complicated distribution stands to improve by at most $\log e$ bits over the current bound.

Theorem~\ref{thm:rts}, along with Theorem III.1 of~\cite{flamich2025redundancy}, gives a complete characterization of the one-shot rate of relative entropy codes for channels with a bounded density ratio:
\begin{corollary}\added{[Tight One-Shot Characterization of Relative Entropy Coding]}
\label{cor:gnsct}
Let $X, Y \sim P_{X, Y}$ be such that $\RND{P_{Y \mid X}}{P_Y}(y \mid x) \leq M$ for all $x$ and $y$. 
Let $H^* \defeq \inf_{Z \perp X} H(Y \mid Z)$ be the minimum rate achievable by any relative entropy code. Then,
\begin{equation}
\FI{X}{Y} \leq H^* \leq \FI{X}{Y} + \log e. \label{eq:sandwich}
\end{equation}
\end{corollary} 
\begin{figure}
    \centering
    \includegraphics[]{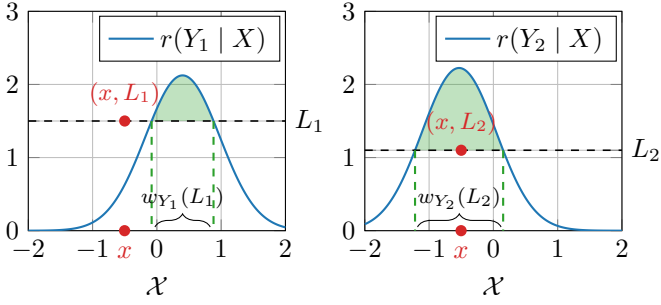}
    \caption{A visualization of the \textcolor{tabgreen}{acceptance region} in rejection sampling for an additive white Gaussian noise channel $Y = X + N$ with $X \sim \Normal{0}{0.75}$ and $N \sim \Normal{0}{0.25}$, where $X$ and $Y$ are truncated to a large but finite range to ensure a bounded density ratio. Here, \textcolor{tabred}{$x$}$\;=-0.5$ and the two proposal samples are $Y_1 = 0.3, Y_2=-0.4$ with heights $L_1=1.5$ and $L_2=1.1$, where we use the shorthand $L_i \defeq MU_i$. We see that  \textcolor{tabred}{$x$} $\not\in S_1$ but \textcolor{tabred}{$x$} $\in S_2$, meaning $Y_1$ is rejected and $Y_2$ is accepted. The $P_X$-probability of hitting the set $S_i$ is exactly $w_{Y_i}(L_i)$.}
    \label{fig:rts}
\end{figure}
Note that the assumption of Corollary~\ref{cor:gnsct} applies to all finitely-supported channels. 
Furthermore, when $X$ is discrete and $X \to Y$ is the trivial channel $Y = X$, we have $\FI{X}{Y}=\Ent(X)$.
Thus, the corollary also recovers the noiseless source coding theorem with a slightly worse constant term of $1 + \log e \approx 2.44$ instead of $1$.
The cost of the ring toss code can also be written in terms of the mutual information, improving the best-known constant from 2.74 to 2.45~\cite{li2025discrete}.
\begin{corollary}
Let $P_{X, Y}$ be as in Corollary~\ref{cor:gnsct}. Then, the ring toss code satisfies
\begin{align*}
\Ent(Y \mid Z) 
&\leq \MI{X}{Y} + \log (\MI{X}{Y} + 1) + 1 + \log e \\
&<\MI{X}{Y} + \log (\MI{X}{Y} + 1) + 2.45.
\end{align*}
\end{corollary}
\begin{proof}
The result is immediate after applying the upper bound in~\eqref{eq:kltocsd} to the result of Theorem~\ref{thm:rts} and invoking Jensen's inequality. 
\end{proof}
\section{Singular Relative Entropy Coding} \label{sec:singular}
The ring toss code applies particularly naturally to so-called \textit{singular channels}, which we present next.
\begin{definition}[Singular Channel~\cite{altuug2014refinement, sriramu2024optimal}] \label{def:singular}
Given $P_{X, Y}$, the channel $X \to Y$ is singular if there exists a $P_Y$-measurable function $g$ such that
\begin{equation}
\RND{P_{Y \mid X}}{P_Y}(Y \mid X) = g(Y) \quad P_{X, Y}\text{-almost surely}. 
\end{equation}
\end{definition}
Any channel that is not singular is called \emph{nonsingular}. The two canonical examples of singular channels are the binary erasure channel and the continuous uniform additive noise channel. Singular channels are important as they have been shown to admit relative entropy codes with sub-logarithmic redundancy in the asymptotic case~\cite{sriramu2024optimal, flamich2026singular}. 
Using Theorem~\ref{thm:rts}, we obtain a significantly stronger statement: the ring toss code for singular channels achieves constant one-shot redundancy. 
\begin{corollary} \label{cor:noredund}
Let $X, Y \sim P_{X, Y}$ with $X \to Y$ singular and $\RND{P_{Y \mid X}}{P_Y}(y \mid x) \leq M$ for all $x$ and $y$.
The ring toss code achieves
\begin{equation}
\Ent(Y \mid Z) \leq \MI{X}{Y} + \log e.
\end{equation}
\end{corollary}
\begin{proof}
We begin by proving the following lemma.
\begin{lemma}[A characterization of channel singularity]
\label{lemma:singularity}
Let $X, Y \sim P_{X, Y}$ be a pair of dependent random variables. Then, the following statements are equivalent:
\begin{enumerate}
\item The channel $X \to Y$ is singular.
\item $\MI{X}{Y} = \FI{X}{Y}$.
\end{enumerate}
\end{lemma}
\begin{proof} 
For $P \ll Q$, the proof of Lemma 3.1.2 in~\cite{flamich2024data} shows that $\CSD{P}{Q} \geq \KLD{P}{Q}$, with equality if and only if $w_P(\ell) = 1$ for all $\ell$ with $w_Q(\ell) > 0$.
Define $g \defeq \sup \qty{\ell \in [0, \infty) : w_Q(\ell) > 0}$ and write $\!r = dP / dQ\!$. The equality condition implies that $r \geq g$ $P$-almost everywhere. Moreover, the definition of $g$ means $r \leq g$ $Q$-almost everywhere, and because $P \ll Q$, $r \leq g$ $P$-almost everywhere too. Thus, equality holds if and only if $r = g$ $P$-almost everywhere. For $X, Y$ random-variables,  fixing $Y = y$ and letting $Q = P_X$ and $P = P_{X = Y \mid y}$, the equality condition becomes 
\begin{equation}
\RND{P_{X \mid Y}}{P_X} (X \mid y) = g \quad P_{X \mid Y = y}\text{-almost surely}. \label{eq:equalitysingular}
\end{equation}
To have $\FI{X}{Y} = I(X;Y)$, equality must hold in~\eqref{eq:equalitysingular} for each $y$ with $P_Y(y) > 0$. That is,
\begin{equation}
\RND{P_{X \mid Y}}{P_X}(X \mid Y) = g(Y) \quad P_{X, Y}\text{-almost surely}
\end{equation}
for some $P_Y$-measurable function $g$. Hence, $\FI{X}{Y} = \MI{X}{Y}$ if and only if $X \to Y$ is singular.   
\end{proof}
Now, the result follows from \Cref{thm:rts} and Lemma~\ref{lemma:singularity}.
\end{proof}
The ring toss code for singular channels admits an equivalent description, visualized in Fig.~\ref{fig:singular_ring_toss}, which further illustrates how the problem has been transformed into a \emph{search over~$\XSpace$}. 
Specifically, for singular $X, Y \sim P_{X, Y}$ with bounded density ratio $\norm{g}_\infty \leq M$, and common randomness $Z = \qty{(Y_i, U_i)}_{i \geq 1} \overset{\text{i.i.d.}}{\sim} P_Y \times \Unif(0, 1)$, define the set 
\begin{equation}
\label{eq:singular_rtc_filtered_indexes}
\filteredInds = \qty{i \in \Nats : MU_i \leq  g(Y_i)}.
\end{equation}
Notice that $\filteredInds$ is exactly the set of all indices whose corresponding sample $Y_i$ \emph{could} be accepted by rejection sampling. However, recall that although $r = \RND{P_{Y \mid X}}{P_Y}$ is constant over its support, it might not be supported on the whole of $\XSpace$. Thus, not every index in $\mathcal{S}$ would be accepted by rejection sampling, as any sample with $P_{Y \mid X}(Y_i \mid x) = 0$ will be rejected. This probability cannot be computed by the decoder, which lacks knowledge of~$x$. Therefore, the encoder computes and encodes 
\begin{equation}
    K = \min \qty{i \in \filteredInds : \RND{P_{Y \mid X}}{P_Y}(Y_i \mid x) > 0}. \label{eq:ksingular}
\end{equation}
Inspecting~\eqref{eq:ksingular} in light of the proof of Lemma~\ref{lemma:singularity}, it is clear that $K$ is again exactly the index returned by rejection sampling. Written this way, the encoder and decoder both ``pre-filter" the sequence of proposal samples for those which could be accepted, and the encoder communicates the first $Y_i,$ $i \in \mathcal{S}$ which has a nonzero probability of being generated by $x$. This pre-filtering is implicitly built into $Q_{K \mid Z}$, which assigns zero probability to any $k$ for which $U_k$ exceeds ${\sup_x r(Y_k \mid x)/M}$. This interpretation is the namesake of the \emph{ring toss code}, as the decoder ``tosses'' the support $S_i$ of the function $x \mapsto r(Y_i \mid x)$, attempting to land it on $x$, akin to the game of the same name.  
\section{Recovering the Asymptotic Characterization of Relative Entropy Coding}
\begin{figure}
    \centering
    \vspace*{0.04in}
    \includegraphics[]{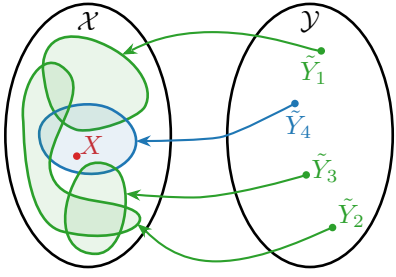}
    \caption{Illustration of the ring toss code for $X, Y \sim P_{X, Y}$ when the channel $X \to Y$ is singular. 
    Given the common randomness $Z$, upon \deleted{the sender} observing the source {\color{tabred} $X \sim P_X$}, the \added{sender} compute\added{s} the filtered set of indices $\filteredInds$ from \eqref{eq:singular_rtc_filtered_indexes}.
    Denote the the filtered samples {\color{tabgreen} $\{\tilde{Y}_i\}_{i \geq 1} = \qty{Y_i}_{i \in \filteredInds}$}.
    Now, for each $\tilde{Y}_i$, the encoder computes the support $S_i$ (shaded {\color{tabgreen} green} and {\color{tabblue} blue} subsets of $\XSpace$ in the figure) of the function $x \mapsto \frac{dP_{Y \mid X}}{dP_Y}(\tilde{Y_i} \mid x)$ and accepts the first $\tilde{Y}_i$ where $X \in S_i$, {\color{tabblue} $\tilde{Y}_4$} in the example above.
    }
    \label{fig:singular_ring_toss}
\end{figure}
Together, Theorem~\ref{thm:rts} and Corollary~\ref{cor:noredund} immediately recover previous results on the asymptotic rate of relative entropy coding ~\cite{sriramu2024optimal,flamich2025redundancy}. 
Concretely, given $P_{X, Y}$, consider the problem of simulating $n$ independent copies of the channel $X \to Y$, i.e., letting the proposal distribution be $P_{Y}^{\otimes n}$ and, upon input $X^n \sim P_{X}^{\otimes n}$, sending a binary message so the encoder can generate a sample $Y^n \sim P_{Y^n \mid X^n = x^n}$. In this problem, of interest is the best achievable rate per dimension
\begin{equation}
    R_n \defeq \tfrac{1}{n} \inf\nolimits_{(Z, \enc,\dec)} \Exp[ \lvert \enc(X^n, Z) \rvert].
\end{equation}
The distribution $P_{X^n, Y^n}$ satisfies ${\MI{X^n}{Y^n} = n \MI{X}{Y}}$, and hence ${R_n \to \MI{X}{Y}}$ as ${n \to \infty}$. Given the fundamental lower bound ${\Ent(Y \mid Z) \geq \MI{X}{Y}}$, this is optimal. 
In light of the logarithmic gap over the mutual information in the one-shot case,~\cite{sriramu2024optimal,flamich2025redundancy} study the logarithmic redundancy 
\begin{equation}
    R^{\text{log}} \defeq \lim_{n \to \infty} \frac{n (R_n - \MI{X}{Y})}{\log n}.
\end{equation}
It is clear from inspection that $0 \leq R^{\text{log}} \leq 1$. As a corollary to Theorem~\ref{thm:rts}, we obtain the same characterization of $R^{\text{log}}$ as~\cite{flamich2025redundancy}:
\begin{corollary}
    Under the asymptotic setup previously described, for any channel $X \to Y$ with bounded density ratio, 
    \begin{equation}
        R^{\textup{log}} = \begin{cases}
            0, &X \to Y \text{ singular} \\
            1/2 &X \to Y \text{ nonsingular}. 
        \end{cases}
    \end{equation}
\end{corollary}
\begin{proof}
The singular case is immediate from Corollary~\ref{cor:noredund}. 
For the nonsingular case, we use a result from the proof of Theorem 3.1.6 in \cite{flamich2025redundancy}, which shows that nonsingular channels whose information density has a finite second moment satisfy
\begin{equation}
\lim_{n \to \infty} \frac{\FI{X^n}{Y^n} - \MI{X^n}{Y^n}}{\log n} = \frac{1}{2}.
\end{equation}
Since we assumed that the density ratio of $X \to Y$ is uniformly bounded, the above result applies; combining it with \Cref{thm:rts} completes the proof.
\end{proof}
Our construction has the advantage over~\cite{sriramu2024optimal, flamich2026singular} of being significantly simpler and one-shot optimal. In the singular case, the achievability algorithms of both~\cite{sriramu2024optimal} and~\cite{flamich2026singular} incur an $\Oh(\log \log n)$ penalty over $\MI{X}{Y}$. 

\section{Discussion} \label{sec:discussion}
We have defined the \emph{ring toss code}, an encoding method for standard rejection sampling, and proved it is optimal for relative entropy coding. In doing so, we provide a matching achievability bound to the one-shot lower bound in Theorem III.1 of~\cite {flamich2025redundancy}. For singular channels, our upper bound achieves a constant redundancy over the mutual information. 
\par
There are several interesting directions of future work stemming from our results. Most immediately, it is of interest to remove the boundedness condition $\norm{\RND{P_{Y \mid X}}{P_Y}}_\infty \leq M$. We emphasize that this is a requirement of the rejection sampling algorithm, and are optimistic that the insight of encoding the probability of hitting a set in $\XSpace$-space can be adapted to other relative entropy codes \added{with fewer assumptions and better computational complexity}. Moreover,~\cite{goc2024channel} shows that any sampling algorithm will have expected runtime ${\Exp[K] \geq \norm{\RND{P_{Y \mid X}}{P_Y}}_\infty}\!\!$.
Thus, the boundedness assumption is equivalent to the algorithm terminating almost surely. 
\par
The functional information also has several interesting open questions about its properties.
For instance, are there classes other than singular channels for which the gap between $\FI{X}{Y}$ and $\MI{X}{Y}$ is sub-logarithmic, and more practically, can $\FI{X}{Y}$ be efficiently computed when $X$ and $Y$ are Gaussian random variables?
\appendices
\crefalias{section}{appendix}
\section{Proof that $Q_{K \mid Z}$ is a probability distribution} \label{app:dist}
Let $X, Y, M$, $Z = \qty{(Y_i, U_i)}_{i \geq 1}$, and $Q_{K \mid Z}$ be as in \Cref{thm:rts}.
We now show that for $P_Z$-almost every $Z = z$,
\begin{equation}
Q_{K \mid Z}(k \mid z) = w_{y_k}( Mu_k) \prod_{i=1}^{k-1} (1 - w_{y_i}(Mu_i))
\end{equation}
is a valid probability distribution. From the fact that $w_y$ is a probability density function, we obtain that ${0 \leq Q_{K \mid Z}(k \mid z) \leq 1}$ for each $k$. Now, we wish to show that $Q_{K \mid Z}(k \mid z)$ sums to 1. We calculate that
\begin{align*}
    &\sum_{k=1}^n Q_{K \mid Z}(k \mid z) \\
    &= \sum_{k=1}^n \left( \prod_{i=1}^{k-1} (1 - w_{y_i}(Mu_i)) - \prod_{i=1}^k (1 - w_{y_i}(Mu_i)) \right) \\
    &= 1 - \prod_{i=1}^n (1 - w_{y_i}(Mu_i)),
\end{align*}
where $\prod_{i=1}^0 (1 - w_{y_i}(Mu_i)) \defeq 1$. We obtain the result after concluding that $\lim_{n \to \infty} \prod_{i=1}^n (1 - w_{y_i}(Mu_i)) = 0$ almost surely, as the $w_{Y_i}(MU_i)$ are i.i.d. with mean greater than 0. 

\balance 
\bibliographystyle{ieeetr}
\bibliography{references}

\end{document}